\documentclass[runningheads]{llncs}

\usepackage{bbding}

\usepackage{graphicx}
\usepackage{epsfig}

\usepackage{comment}

\usepackage{amssymb}

\usepackage{amsmath}
\usepackage{subfigure}

\usepackage{url}

\usepackage{multirow}

 \usepackage[normalem]{ulem}
 \useunder{\uline}{\ul}{}

\usepackage{enumitem}

\usepackage{tikz}
\usetikzlibrary{matrix}

\usepackage{xspace}
\usepackage{todonotes}

\begin{document}

\def\BWT{$BWT$}
\def\RLE{$RLE$}

\def\cbwt(#1){\texttt{c-bwt}(#1)}
\def\cbwt{\texttt{cbwt}}

\def\bwt(#1){\texttt{bwt}(#1)}
\def\bwt{\texttt{bwt}}

\begin{frontmatter}

\title{String Attractors and Combinatorics on Words}



\newcommand{\repeatthanks}{\textsuperscript{\thefootnote}}

\author{Sabrina Mantaci\inst{1} \and Antonio Restivo\inst{1} \and
Giuseppe Romana\inst{1} \and
Giovanna Rosone\inst{2}\Envelope
\and Marinella Sciortino\inst{1}\Envelope
}
\authorrunning{S. Mantaci, A. Restivo, G. Romana, G. Rosone and M. Sciortino} 

\institute{University of Palermo, Italy,\\
\email{sabrina.mantaci@unipa.it, antonio.restivo@unipa.it, giuseppe.romana01@community.unipa.it, marinella.sciortino@unipa.it}\\
\and
University of Pisa, Italy,\\
\email{giovanna.rosone@unipi.it}}



\maketitle   

\begin{abstract}
The notion of \emph{string attractor} has recently been introduced in [Prezza, 2017] and studied in [Kempa and Prezza, 2018] to provide a unifying framework for known dictionary-based compressors. A string attractor for a word $w=w[1]w[2]\cdots w[n]$ is a subset $\Gamma$ of the positions $\{1,\ldots,n\}$, such that all distinct factors of $w$ have an occurrence crossing at least one of the elements of $\Gamma$. While finding the smallest string attractor for a word is a NP-complete problem, it has been proved in [Kempa and Prezza, 2018] that dictionary compressors can be interpreted as algorithms approximating the smallest string attractor for a given word.

In this paper we explore the notion of string attractor from a combinatorial point of view, by focusing on several families of finite words. The results presented in the paper suggest that the notion of string attractor can be used to define new tools to investigate combinatorial properties of the words.

\end{abstract}
\end{frontmatter}%
\section{Introduction}
The notion of \emph{String Attractor} has been recently introduced and studied in \cite{prezzaArXiv2017,KempaP18} to find a common principle underlying the main techniques constituting the fields of dictionary-based compression. 
It is a subset of the text's positions such that all distinct factors have an occurrence crossing at least one of the string attractor's elements. 
From one hand the problem of finding the smallest string attractor of a word has been proved to be NP-complete, on the other hand  dictionary compressors can be interpreted as algorithms approximating the smallest string attractor for a given word \cite{KempaP18}. Moreover, approssimation rates with respect to the smallest string attractor can be derived for most known compressors. In particular compressors based on the Burrows-Wheeler Transform and the dictionary-based compressors are considered. 

The Burrows-Wheeler Transform (BWT) is a reversible transformation that was introduced in 1994 in the field of Data Compression and it also largely used for self-indexing data structures. It has several combinatorial properties that make it a versatile tool in several contexts and applications \cite{PrezzaPSR19,LouzaTGZ18,MaReSc,RestivoRosoneTCS2011,MantaciRRSV17}.

Dictionary-based compressors are mainly based on a technique originated in two theoretical papers of Ziv and Lempel \cite{Ziv77auniversal,Ziv78}. Such compressors, that are able to combine compression power and
compression/decompression speed, are based on a paper in which combinatorial properties of word factorization are explored \cite{LZ1976}.

In this paper we explore the notion of string attractor from a combinatorial point of view. In particular, we compute the size of a smallest string attractor for infinite families of words that are well known in the field of Combinatorics on Words: standard Sturmian words (and some their extension to bigger alphabets), Thue-Morse words and de Brujin words. In particular, we show that the size of the smallest string attractor for standard Sturmian words is $2$ and it contains two consecutive positions. For the de Brujin words the size of the smallest string attractor grows asintotically as $\frac{n}{\log n}$, where $n$ is the length of the word. We show a string attractor of size $\log n$ for Thue-Morse words and we conjecture that this size is minimum.  From the results presented in the paper, we believe that the distribution of the position in the smallest string attractor of a word, in addition to its size, can provide some interesting information about the combinatorial properties of the word itself. For this reason the notion of string attractor can provide hints for defining new methods and measures to investigate the combinatorial complexity of the words.

\section{Preliminaries}
Let $\Sigma =\{a_1, a_2, \ldots, a_\sigma\}$ be a finite ordered alphabet with $a_1< a_2< \ldots < a_\sigma$, where $<$ denotes the standard lexicographic order. We denote by $\Sigma^*$ the set of words over $\Sigma$.
Given a finite word $w = w_1w_2 \cdots w_n \in \Sigma^*$ with each $w_i \in \Sigma$, the length of $w$, denoted $|w|$, is equal to $n$.

Given a finite word $w=w_1w_2\cdots w_{n}$ with each $w_i \in \Sigma$, a \emph{factor} of a word $w$ is written as $w[i,j] = w_i \cdots w_j$ with $1\leq i \leq j \leq n$. A factor of type $w[1,j]$ is called a \emph{prefix}, while a factor of type $w[i,n]$ is called a \emph{suffix}.
We also denote by $w[i]$ the $i$-th letter in $w$ for any $1\leq i \leq n$.


We denote by $\tilde{w}$ the reversal of $w$, given by $\tilde{w} = w_n \cdots w_2w_1$.
If $w$ is a word that has the property of reading the same in either direction,
i.e. if $w =\tilde{w}$, then $w$ is called a \emph{palindrome}.


We say that two words $x,y\in \Sigma^*$ are {\em conjugate}, if $x=uv$ and $y=vu$, where $u,v\in \Sigma^*$.
Conjugacy between words is an equivalence relation over $\Sigma^*$.


Given a finite word $w$, $w^k$ denotes the word obtained by concatenating $k$ copies of $w$.  A nonempty word $w \in \Sigma^+$ is {\em primitive} if $w=u^h$ implies $w=u$ and $h=1$.
A word $x$ is {\em periodic} if there exists a positive integer $p$ such that $x[i]=x[j]$ if $i=j\mod p$. The integer $p$ is called \emph{period} of $x$.

A {\em Lyndon word} is a primitive word which is the minimum in its conjugacy class, with respect to the lexicographic order relation. We call \emph{Lyndon conjugate} of a primitive word $w$ the conjugate of $w$ that is a Lyndon word.




The Burrows-Wheeler Transform is a permutation $bwt(v)$ of the symbols in $v$, obtained as the concatenation of the last symbol of each conjugate in the list of the lexicographically sorted conjugates of $v$.


The LZ factorization of a word $w$ is its factorization $s = p_1 \cdots p_z$ built left to right in a greedy way by the following rule: each new factor (also called an LZ77 phrase) $p_i$ is either the leftmost occurrence of a letter in $w$ or the longest prefix of $p_i \cdots p_z$ which occurs, as a factor, in $p_1 \cdots p_{i-1}$.


\section{String Attractor of a word}

In this section we describe the notion of string attractor that is a combinatorial object introduced in 
\cite{prezzaArXiv2017,KempaPrezzaArXiv2017} to obtain a unifying framework for dictionary compressors. 


\begin{definition} 
A {\em string attractor} of a word $w \in \Sigma^n$ is a set of $\gamma$ positions $\Gamma = \{j_1, \ldots, j_\gamma \}$ such
that every factor $w[i, j]$ has an occurrence $w[i',j'] = w[i,j]$ with $j_k \in [i', j']$, for some $j_k \in \Gamma$.
\end{definition}

Simply put, a string attractor for a word $w$ is a set of positions in $w$ such that all distinct factors of $w$ have an occurrence crossing at least one of the attractor’s elements. 

Note that, trivially, any set that contains a string attractor for $w$, is a string attractor for $w$ as well. Note also that a word can have different string attractors that are not included into each other. We are interested in finding a {\em smallest string attractor}, i.e. a string attractor with a minimum number of elements. We denote by $\gamma^*(w)$ the size of the smallest string attractor for $w$. Note that all the factors made of a single letters should be covered, and therefore $\gamma^*(w)\geq |\Sigma|$.

\begin{example}
Let $w=adcbaadcbadc$ be a word on the alphabet $\Sigma=\{a,b,c,d\}$. A string attractor for $w$ is for instance $\Gamma=\{1,4,6,8,11\}$. Note that, in order to have a string attractor,  position $1$ can be removed from $\Gamma$, since all the factors that cross position $1$ have a different occurrence that crosses a different position in $\Gamma$. Therefore $\Gamma'=\{4,6,8,11\}$ is also a string attractor for $w$ with a smaller number of elements. The positions of $\Gamma'$ are underlined in $$w=adc\underline{b}a\underline{a}d\underline{c}ba\underline{d}c.$$ $\Gamma'$ is also a smallest string attractor since $|\Gamma|=|\Sigma|$. Then  $\gamma^*(w)=4$.  
Remark that the sets $\{3,4,5,11\}$ and $\{3,4,6,7,11\}$ are also string attractors for $w$. It is easy to verify that the set $\Delta=\{1,2,3,4\}$ is not a string attractor since the factor $aa$ does not intersect any position in $\Delta$.
\label{exp:gamma}
\end{example}

The following two propositions, proved in \cite{prezzaArXiv2017}, are useful to derive a lower bound on the value of $\gamma^*$.

\begin{proposition}
Let $\Gamma$ be a string attractor for the word $w$. Then, $w$ contains at most $|\Gamma|k$ distinct factors of length $k$, for every $1\leq k \leq |w|$.\label{prop:kfactors}
\end{proposition}

\begin{proposition}\label{prop:repeatedfactor}
Let $w\in \Sigma^*$ and let $r$ be the length of its longest repeated factor. Then it holds $\gamma^*(w)\geq \frac{|w|-r}{r+1}$.
\end{proposition}

The following proposition gives an upper bound for $\gamma^*$ of a concatenation of words, when $\gamma^*$ of the single words are known. 

\begin{proposition}\label{prop:concatenate}
Let $u$ and $v$ two words, then $\gamma^*(uv)\leq \gamma^*(u)+\gamma^*(v)+1$.
\end{proposition}

\begin{example}
The bound defined in the previous proposition is tight. In fact, let $u=\underline{b}aa\underline{a}ba$ and $v=c\underline{d}cc\underline{c}d$ be two words in which the positions of the smallest string attractors are underlined. If we consider $uv=\underline{b}aa\underline{a}ba\underline{c}\underline{d}cc\underline{c}d$, the underlined positions represent one of the smallest string attractors for $uv$, as one can verify.
\end{example}

The following proposition gives an upper and lower bound for $\gamma^*$, when a power of a given word is considered.

\begin{proposition}\label{p-power}
Let $w=u^n$. Then $\gamma^*(u)\leq \gamma^*(u^n)\leq \gamma^*(u)+1$. 

\end{proposition}

\begin{example}
The upper bound given by Proposition \ref{p-power} is tight. In fact consider the word $u=abbaab$. It is easy to check that the only smallest string attractors for $u$ are $\Gamma_1^*=\{2,4\}$ and $\Gamma_2^*=\{3,5\}$. 
In order to find the smallest string attractor for $u^2=abbaababbaab$, we remark that neither $\Gamma_1$ nor $\Gamma_2$ (neither any string attractor obtained from them by moving some position from the first to the second occurrence of $u$) cover all the new factors that appears after the concatenation. 
In particular $aba$ is not covered by $\Gamma_1^{'}$, and $bab$ is not covered by $\Gamma_2'$. A way to get the smallest string attractor for $u^2$ is to add to $\Gamma_1$ or $\Gamma_2$, the position corresponding either to the end of the first occurrence of $u$ or the beginning of the second occurrence. For instance, $\Gamma^*=\{2,4,6\}$ is a smallest string attractor for $u^2$. 
\end{example}

\begin{example}
Remark that $\gamma^*(u^n)$ can be equal to $\gamma^*(u)$ although different point for the string attractor could be chosen. For instance, let $u=a\underline{b}\underline{a}b\underline{c}bc$ be a word whose smallest string attractor is $\{2,3,5\}$ (the underlined letters). Then $u^2= \underline{ab}abcb\underline{c}ababcbc$ has a string attractor $\{3,6,7\}$ of cardinality $3$. Remark that $\{2,3,5\}$ is not a string attractor for $u^2$.

\end{example}
A straightforward consequence of Proposition \ref{p-power} is the following:

\begin{corollary}\label{cor:conjAttractor}
If $u$ and $v$ are conjugate words, then $|\gamma^*(u)-\gamma^*(v)| \leq 1$.
\end{corollary}

\begin{example}
Consider the word $w=babbaaa$. Then a smallest string attractor for $w$ is $\{3,5\}$, i.e. $\gamma^*(w)=2$. Consider its conjugate $u=a\underline{b}ab\underline{b}a\underline{a}$. Its smallest string attractor is $\{2,4,6\}$, i.e. $\gamma^*(u)=3$. Note that the Lyndon word does not have necessarily the smallest $\gamma^*$ among the conjugates. In fact, for instance, the Lyndon conjugate of $w$ is $aaababb$, and it is easy to verify that one of the smallest string attractor is
$\{3,4,6\}$.
\end{example}

\section{Approximating a string attractor via compressors}


In \cite{KempaP18} the authors show that many of the most well-known compression schemes reducing the text’s size by exploiting its repetitiveness can induce string attractors whose sizes are bounded by the repetitiveness measures associated to such compressors. In particular, straight-line programs, Run-Length Burrows-Wheeler transform, macro schemes, collage systems, and the compact directed acyclic word
graph are considered. Here we report some result related to the Burrows-Wheeler Transform, collage systems, and Lempel-Ziv 77 (that is a particular macro-scheme) that provide upper bounds on the size of the smallest string attractor for a given word.  Such bounds will be used in next sections to compute the string attractors for known families of finite words.


The first theorem, proved in \cite{KempaP18}, states a connection between a string attractor of a word $w$ and the runs of equal letters in the $\bwt(w\$)$.
\begin{theorem}
\label{th:prezzaBWT}
Let $w$ be a word and let $r$ be the number of equal-letter runs in the $bwt(w\$)$, where $\$$ is a simbol different from the ones in the alphabet $\Sigma$ and assumed smaller than any symbol in $\Sigma$. 
Then, $w$ has a string attractor of size $r$.
\end{theorem}	

In particular, in the proof of Theorem \ref{th:prezzaBWT} the string attractor is constructed by considering the position of the symbols in $w$ that correspond, in the output of the transformation, to the first occurrence of a symbol in each run (or, equivalently, the last occurrence of a symbol in  each run).

The following result, proved in \cite{KempaP18}, states the relationship between a string attractor of a word $w$ and the number of phrases in the LZ  parsing of $w$.
\begin{theorem}\label{L_KePrLZ}
Given a word $w$, there exist a string attractor of $w$ of size equal to the number of phrases of its LZ parsing. 
\end{theorem}

By using the previous result, a string attractor can be constructed by considering the set of  positions at the end of each phrase. 

The following theorem in \cite{KempaP18} gives a connection between a particular class of grammars, called \emph{collage systems} \cite{KIDA2003253}, and string attractors.

\begin{definition}
A collage system is a set of $c$ rules of four possible types:
\begin{itemize}
    \item $X\to a$: nonterminal $X$
    expands to a terminal $a$
\item $X\to AB$: nonterminal $X$ expands to $AB$, with $A$ and $B$ nonterminals different from $X$
\item $X \to R^{\ell}$: nonterminal $X$ expands to nonterminal $R \neq X$
repeated $\ell$ times
\item $X \to K[l,r]$: nonterminal $X$ expands to a substring of the
expansion of nonterminal $K\neq X$.
\end{itemize}
\end{definition}

\begin{theorem}\label{th-PrezzaGr}
Let $G=\{X_i\to a_i, i=1, \ldots g'\} \cup \{X_i\to A_iB_i, i=1,\ldots, g''\}\cup\{Y_i\to Z_i^{l_i}, l_i\geq 2, i=1,\ldots, g'''\}\cup\{W_i\to K_i[l_i\ldots r_i], i=1\ldots, g''''\}$ be a collage system of size $g=g'+g''+g'''+g''''$ generating a word $w$. Then $w$ has a string attractors of size at most $g$.
\end{theorem}

By using this theorem one can easily find the minimum string attractor for a particular class of very ``regular'' strings, as shown in the following corollary. 


\begin{corollary}
Let $u \in \Sigma^*$ be a word.
If $u$ is the form of $u=\sigma_{i_1}^{n_1}\sigma_{i_2}^{n_2} \cdots \sigma_{i_k}^{n_k}$ (where all $\sigma_{i_j}$ are different symbols in the alphabet), then  $\Gamma=\{n_1,n_1+n_2, \ldots, n_1+\ldots+n_k\}$ is a string attractor of minimum size for $u$, so $\gamma^*(u)=\sigma$.
\end{corollary}


\section{Minimum size string attractors}

In this section, we analyze the words whose smallest string attractor has size equal to the the size of the alphabet, that is the minimum possible size. 
In the following two subsections we distinguish the case of binary words and the case of words over alphabets with more than two letters.

\subsection{Binary words}

In this subsection we focus on an infinite family of finite binary words whose minimum string attractor has size $2$.

Standard Sturmian words is a very well known family of binary words that are the basic bricks used for the construction of infinite Sturmian words, in the  sense  that  every  characteristic Sturmian word is the limit of a sequence of standard words (cf. Chapter 2 of \cite{Lothaire:2005}). 
These words have a multitude of characterizations and appear as extremal case in a very great range of contexts \cite{KMP77,CS09}. In this paper two of their characterization are particularly useful, that is a special decomposition into palindrome words and an extremal property on the periods of the word that is closely related to Fine and Wilf's theorem (cf. \cite{deLucaMignosi1994,deLuca1997}). More formally, standard Sturmian words can be defined in the following way which is a natural generalization of the definition of the Fibonacci word. Let $q_0,q_1,\ldots q_n,\ldots$ any sequence of natural integers such that $q_0 \geq 0$ and $q_i > 0$ ($i = 1,\ldots,n$), called \emph{directive sequence}. The sequence $\{s_n\}_{n\geq 0}$ can be defined inductively as follows: $s_0 = b$, $s_1=a$, $s_{n+1}= (s_{n})^{q_{n-1}}s_{n-1}$, for $n > 1$. We denote by $Stand$ the set of all words $s_n$, $n\geq 0$, constructed for any directive sequence of integers.

Furthermore, another characterization of standard Sturmian words is related to the Burrows Wheeler Transform (BWT) since, for binary alphabets, the application of the BWT to standard Sturmian words produces a total clustering of all the instances of any character (cf. \cite{MaReSc}), as reported in the following theorem.

\begin{theorem}[\cite{MaReSc}]\label{th_cluster}
Let $w \in \Sigma^*$. Then $w$ is a conjugate of a word in $Stand$ if and only if $bwt(w) = b^pa^q$ with $gcd(p,q)=1$.
\end{theorem}

In the following theorem, for each standard Sturmian word, we individuate a string attractor, whose positions are strictly related with particular decompositions of such words depending on their periodicity. In particular, we recall that $Stand=\{a,b\}\cup PER\{ab,ba\}$ (cf. \cite{deLucaMignosi1994}), where $PER$ is the set of all words $v$ having two periods $p$ and $q$ such that $gcd(p,q)=1$ and $|v| = p + q - 2$. Given a word $w\in Stand$, we denote by $\pi(w)$ its prefix of length $|w|-2$, belonging to the set $PER$, uniquely defined by using previous equality.  By using a property of words in $PER$ (cf. \cite{deLucaMignosi1994}), $\pi(w)=QxyP=PyxQ$, where $x\neq y$ are characters and $Q$ and $P$ are uniquely determined palindromes. So, a standard Sturmian word $w=\pi(w)ba$ can be decomposed as $w=QxyPba=PyxQba$. We call $PER$ decompositions such  factorizations of $w$.

\begin{theorem}\label{th:string_attract} 
By using $PER$ decompositions, a Standard sturmian word can be decomposed as $w=QxyPba=PyxQba$.
For each $w\in Stand$ with $|w|\geq 2$,  let $\eta$ be the length of the longest palindromic proper prefix of $\pi(w)$, the set $\Gamma_1=\{\eta+1,\eta+2\}$ or the set $\Gamma_2=\{|w|-\eta-3,|w|-\eta-2\}$ is a minimum string attractor for $w$.  
\end{theorem}

\begin{proof}
Let us suppose that $w=\pi(w)ba$.  By using a property of words in $PER$ (cf. \cite{deLucaMignosi1994}), $\pi(w)=QxyP=PyxQ$, where $x\neq y$ are characters and $Q$ and $P$ are uniquely determined palindromes. Let us suppose that $|Q|>|P|$. So, $\eta=|Q|$. Firstly we suppose that $x=b$. This means that $w=QbaPba=PabQba$. From a result in \cite{BersteldeLuca97} $aPabQb$ and $bQbaPa$ are the smallest and the greatest conjugates in the lexicographic order, respectively. By Theorem \ref{th_cluster} and by using an argument similar to the proof of Theorem \ref{th:prezzaBWT}, a string attractor can be constructed by considering the positions corresponding to the end of each run. It is possible to see that such positions correspond to the two characters following the prefix $P$ of length $|w|-\eta-4$. If $x=a$, then $w=QabPba=PbaQba$. In this case $aQabPb$ and $bPbaQa$ are the smallest and the greatest conjugates in the lexicographic order, respectively. In this case the ending positions of each run in the output of $BWT$ correspond to the two characters following the prefix $Q$. So, the positions in the string attractor are $\{\eta+1, \eta+2\}$. The case $w=\pi(w)ab$ can be proved analogously by considering the starting characters of each run in the clustered output of $BWT$. \qed
\end{proof}

\begin{example}
Given the standard Sturmian word $w=ababaababa ab aba ba$, the $PER$ decompositions of $w$ are $ababaababa.ab.aba.ba=aba.ba.ababaababa.ba$, then $\{11,12\}$ is a (smallest) string attractor for $w$, since $\eta=10$. Given $v=abaababaababa$, its $PER$ decompositions are $abaaba.ba.aba.ba=aba.ab.abaaba.ba$ and $\eta=6$. So, $\{4,5\}$ is a string attractor.

\end{example}

Previous theorem shows an infinite family of finite binary words such that the size of the smallest string attractor is minimum. In Example \ref{ex:minimum2} we provide some binary words not belonging to $Stand$ with minimum size of string attractor. An open question is to characterize all the binary words with string attractor of size $2$. Furthermore, we can remark that for the standard Sturmian words one can construct string attractors whose positions are consecutive. This fact could be related to the number of distinct factors appearing in the words. It could be interesting to investigate which classes of binary words have minimum size string attractors containing consecutive positions.

\begin{example}\label{ex:minimum2}
A possible set of smallest string attractor of the words $u=a^nb^m$ or $w=b^na^m$ is $\{n,n+m\}$. Moreover, if we consider the word $u=(ab)^{n_1}(ba)^{n_2}$ or $u=(ba)^{n_1}(ab)^{n_2}$ a smallest string attractor is
$\{2{n_1},2{n_1}+2{n_2}\}$.
\end{example}

\subsection{The case of bigger alphabets}
We now analyze the string attractors for words of more of two letters by generalizing the standard Sturmian words.
Numerous generalizations of Sturmian sequences have been introduced for an alphabet with more than $2$ letters. Among them, one natural generalization are the episturmian sequences that are defined by using the palindromic closure property of Sturmian sequences (cf. \cite{PAQUIN2009}). 
Here we consider some special prefixes of episturmian sequences that are balanced, called \emph{finite epistandard words} \cite{Paquin2007}. We remark that, a word $w$ is balanced if, for any symbol $a$, the numbers of $a$'s in two factors of $w$ of the same length differ at most by $1$ and it is circularly balanced if each conjugate is balanced.
In the case of a binary alphabet they correspond to the standard Sturmian words.

We focus on the circularly balanced epistandard words defined in  Theorem~\ref{Theo:BalEpistClusterized} (see \cite{Paquin2007,RestivoRosoneTCS2011}), that can be built via the \emph{iterated palindromic closure} function. The \emph{iterated palindromic closure} function \cite{Justin2005}, denoted by $Pal$, is defined recursively as follows. Set $Pal(\varepsilon) = \varepsilon$ and, for any word $w$ and letter $x$, define $Pal(wx)~=~(Pal(w)x)^{(+)}$, where $w^{(+)}$, the \emph{palindromic right-closure} of $w$, is the (unique) shortest palindrome having $w$ as a prefix (see \cite{deLuca1997}).
Circularly balanced epistandard words (up to letter permutation), like the standard Sturmian words, are perfectly clustered words under application of the BWT, i.e. the BWT produces a new word that has the minimum number of clusters (\cite{RestivoRosoneTCS2011}, see also \cite{puglisiSimpson2008,RestivoRosoneTCS2009} for more details about the perfectly clustered words on more letters).


\begin{theorem}\label{Theo:BalEpistClusterized}
 Any 
circularly balanced finite epistandard word $t$
belongs to one of the following three families (up to letter permutation):
\begin{enumerate}[label=(\roman*)]
\item \label{type1}   $t = p a_2$, with  $p = Pal(a_1^m a_k a_{k-1}\cdots a_3)$,  where $k\geq 3$ and $m\geq 1$;
\item \label{type2}  $t =  p a_2$, with  $p = Pal(a_1 a_k a_{k-1} \cdots a_{k-\ell} a_1 a_{k-\ell-1}a_{k-\ell-2} \cdots a_3)$, where $0 \leq \ell \leq k -4$ and $k\geq 4$;
\item \label{type3} $t = Pal(a_1 a_k a_{k-1} \cdots a_2)$,  where $k\geq 3$.
\end{enumerate}
\end{theorem}
We observe that the words of the last family of Theorem~\ref{Theo:BalEpistClusterized}  correspond to the {\em Fraenkel's sequence}, that are words related to a important conjecture \cite{FRAENKEL1973}.

For each epistandard word, we find a possible string attractor and show that its size is $\sigma$.

\begin{theorem}\label{th:costantsizeMoreTwo}
If $w$ is a \emph{circularly balanced  epistandard} words, then the \emph{minimum size of string attractor} of $w$ is $\sigma$.
\end{theorem}

The proof of the theorem uses similar arguments as in Theorem \ref{th:string_attract} and the notion of palindromic closure. It will be detailed in the full version of the paper.

\begin{example}
We consider some circularly balanced epistandard words and their corresponding Lyndon words.
A smallest string attractor of
\begin{itemize}
    \item (type \ref{type1}) $w\,=\,(aaa\underline{a}\underline{d}aaaa\underline{c}aaaadaaaa)\underline{b}$, obtained by $p\,=\,Pal(a^4dc)$, is $\{4, 5, 10, 20 \}$,  whereas for  its Lyndon conjugate $aaaa\underline{b}aaa\underline{a}\underline{d}aaaa\underline{c}aaaad$ it is $\{5, 9, 10, 15 \}$.
    \item (type \ref{type2} for $k=4$ and $\ell=0$) $u=(a\underline{d}a\underline{c}ada\underline{a}dacada)\underline{b}$, obtained by $p=Pal(adca)$, is $\{2, 4, 8, 15 \}$ and of its Lyndon conjugate $a\underline{a}\underline{d}a\underline{c}ada\underline{b}adacad$ is $\{2, 3, 5, 9 \}$.
    \item (type \ref{type2} for $k=5$ and $\ell=0$) $u'=(a\underline{e}a\underline{a}ea\underline{d}aeaaea\underline{c}aeaaeadaeaaea)\underline{b}$, obtained by $p=Pal(aeadc)$, is $\{2, 4, 7, 14, 28 \}$ and of its Lyndon conjugate  $aaea\underline{b}aea\underline{a}\underline{e}a\underline{d}aeaaea\underline{c}aeaaeadae$ is $\{5, 9, 10, 12, 19\}$
    \item (type \ref{type2} for $k=5$ and $\ell=1$) $u''=(a\underline{e}a\underline{d}aea\underline{a}eadaea\underline{c}aeadaeaaeadaea)\underline{b}$, obtained by $p=Pal(aedac)$, is $\{2, 4, 8, 15, 30\}$ and of its Lyndon conjugate $aaeadaea\underline{b}aeadaea\underline{a}\underline{e}a\underline{d}aea\underline{c}aeadae$ is $\{9, 17, 18, 20,  24 \}$.
    \item (type \ref{type3}) $v=c\underline{a}\underline{c}\underline{b}cac$, obtained by $p=Pal(cab)$, is $\{2,3,4\}$ and of its Lyndon conjugate 
    $acbcacc$ is $\{3, 5, 6\}$.
    \end{itemize}
Note that, in this example, 
one has an attractor for each symbol $a$ of the alphabet and the position of such attractor in the epistandard word coincides with the position where the last occurrence letter $a \in p$ appears during the palindromic right-closure.
Note also that, in the case of Lyndon words, for each letter $a$, one has an attractor at the position of the last occurrence of $a$ in the run of the output of the BWT.
 \end{example}\label{ex:epist}



Unfortunately, the authors in \cite{RestivoRosoneTCS2011} show that there exist words that do not belong to the families in Theorem~\ref{Theo:BalEpistClusterized} that are perfectly clusterized via BWT, for instance the perfectly clustered word $u = a\underline{b}bbbb\underline{a}\underline{c}ac$ is not a finite epistandard word and a possible string attractor is $\{2,7,8\}$, whereas the perfectly clustered word $v = aacaabaac$ is a finite epistandard but it is not a balanced and a possible string attractor is $\{2,3,6\}$.
The characterization of all perfectly clustered word via BWT is out the scope of this paper.

It remains open the problem of characterizing all words on an alphabet on more of two letters that have the smallest string attractor having size equal to the cardinality of the alphabet.

\section{String Attractors in Thue-Morse Words}

In this section we consider the problem of finding a smallest string attractor for the family of finite binary Thue-Morse words. Thue-Morse words are a sequence of words obtained by the iterated application of a morphism as described below. 

\begin{definition}
Let us consider the alphabet $\Sigma=\{a,b\}$ and the morphism $\varphi:\Sigma^*\mapsto \Sigma^*$ such that $\varphi(a)=ab$ and $\varphi(b)=ba$. Let us denote by $t_n=\varphi^n(a)$ the $n$-th iterate of the morphism $\varphi$ that is called the $n$-th Thue-Morse word.
\end{definition}

It is easy to verify that, for each $n\geq 1$ the $n$-th Thue Morse word has length $2^n$.  The $n$-th Thue-Morse words for $n=3,4,5$ can be found in Figure \ref{fig:addmove}.




By using a result in \cite{Brlek89} on the enumeration of factors in Thue-Morse words the following lower bound is proved.

\begin{proposition}
Let $t_n=\varphi^n(a)$ be the $n$-th Thue-Morse word with $n>2$. Then $\gamma^*(t_n)\geq 3$. 
\end{proposition}



The following proposition provides the recursive structure of the $n$-th Thue-Morse words by using the rules of a context-free grammar.

\begin{proposition}\label{pr-ThMo}
The $n$-th Thue-Morse word is obtained by the following grammar:
$$\{A_0\to a, B_0 \to  b\}\bigcup_{i=1}^{n-1}\{ A_i \to A_{i-1} B_{i-1}, B_i \to B_{i-1} A_{i-1}\}\bigcup\{A_n\to A_{n-1}B_{n-1}\}$$
by taking as axiom the non terminal symbol $A_n$.
\end{proposition}

The grammar described in Proposition \ref{pr-ThMo} contains $2n+1$ rules for the Thue Morse word of length $2^n$. Therefore, by using a result proved in \cite{KempaP18}, reported here as Theorem \ref{th-PrezzaGr}, it is possible to construct a string attractor for $t_n$ having size $2n+1$.

In the last part of this section we exhibit a string attractor $\Gamma_n$ for $t_n$ of size $n$. Our conjecture is that $\gamma^*(t_n)=n$.

\begin{theorem}\label{th:string_attractorTM}
A string attractor of the $n$-th Thue Morse word, with $n\geq 3$ is $$\Gamma_n=\{2^{n-1}+1\}\bigcup_{i=2}^n\{3\cdot 2^{i-2}\}$$
\end{theorem}

\begin{lemma}\label{lem:trequarti}
Let $x$ be a factor of $t_n=\varphi^n(a)$, with $n\geq3$. Then, $t_{n+1}$ admits an occurrence of $x$ crossing a position in the set $\bigcup_{i=2}^{n+1}\{3\cdot 2^{n+1-i}\}$. 
\end{lemma}

\begin{figure}[t!]
    \centering
    \includegraphics[width=\textwidth]{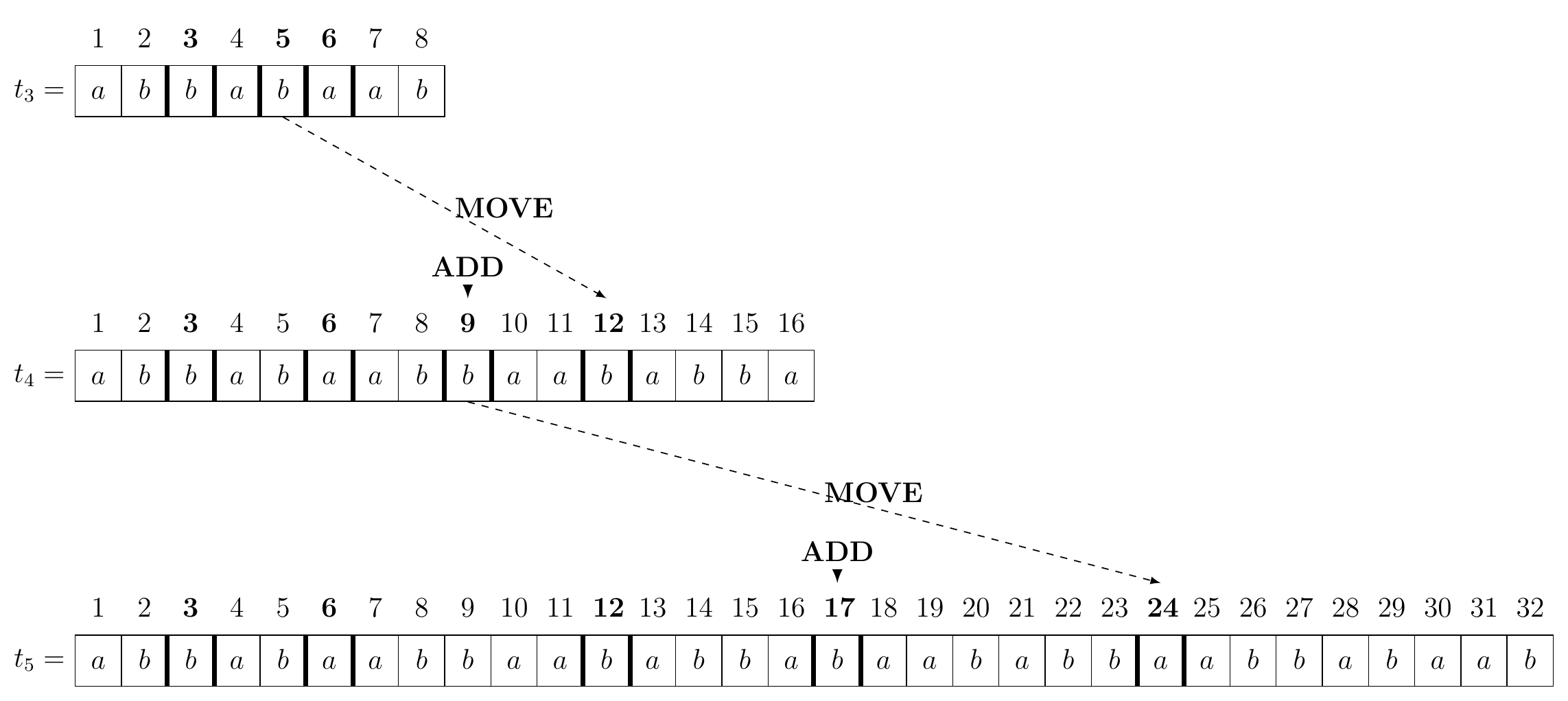}
    \caption{String attractors $\Gamma_n$ for the word $t_n=\varphi^n(a)$, with $n=3,4,5$ (the positions in $\Gamma_n$ are in bold), i.e. $\Gamma_3=\{3,5,6\}$, $\Gamma_4=\{3,6,9,12\}$, $\Gamma_5=\{3,6,12,17,24\}$. Note that in $\Gamma_4$ the position $9$ is obtained by {\sc ADD($9$)} operation, the position $12$ is obtained by the operation {\sc MOVE($5,3\cdot 4$)}. In $\Gamma_5$ the position $17$ is obtained by {\sc ADD($17$)} operation, the position $24$ is obtained by the operation {\sc MOVE($9,3\cdot 8$)}.}
    \label{fig:addmove}
\end{figure}


\begin{proof}[Theorem \ref{th:string_attractorTM}]
The thesis is proved by induction on $n$. If $n=3$, it is easy to check (see Fig. \ref{fig:addmove}) that $\Gamma_3=\{3,5,6\}$ is a string attractor for $t_3$. 

Let us suppose that $\Gamma_{n}$ is a string attractor for $t_{n}$. We show that a string attractor for $t_{n+1}$ can be obtained by applying to $\Gamma_n$ the following two operations:
\begin{itemize}
    \item {\sc ADD($2^{n}+1$)}, that adds the new position $2^{n}+1$
    \item {\sc MOVE($2^{n-1}+1$,$3\cdot 2^{n-1}$)} that replaces the position $2^{n-1}+1$ with $3\cdot 2^{n-1}$.
\end{itemize}
Such operations are described, for $n=3,4,5$, in Fig.  \ref{fig:addmove}.
Let $x$ be a factor of $t_{n+1}$. If $x$ is also a factor of $t_n$, then by Lemma \ref{lem:trequarti} has at least an occurrence crossing a position in the set 
$\Delta=\bigcup_{i=2}^{n+1}\{3\cdot 2^{n+1-i}\}$. We can suppose that $x$ is factor of $t_{n+1}$ that does not appear in $t_n$ and that does not cross any position in $\Delta$. It means that $x$ has to be factor of $u_n v_{n+1}$. In particular, such a factor exists. In fact, let $a$ and $b$ be the first and the last character of $u_n$, then $av_nb$ has only one occurrence in $t_{n+1}$. If follows from the fact that $t_{n+1}=u_n v_n v_n u_n v_n u_n u_n v_n$ and that the Thue-Morse words has no overlapping factors. Moreover, the occurrence of $av_n b$ crosses the position $2^{n}+1$.\qed
\end{proof}

\section{Attractors in de Brujin words}

A {\em de Bruijn} sequence (or words) $B$ of order $k$ on an alphabet $\Sigma$ of size $\sigma$, is a circular sequence in which every possible length-$k$ string on $\Sigma$ occurs exactly once as a substring.

De Brujin words are widely studied in combinatorics on words, and all of them can be constructed by considering all the Eulerian walks on de Brujin graphs. All the de Brujin sequences of order $k$ over an alphabet of size $\sigma$ have length $\sigma^k$.
For instance the (circular) word $w=aaaababbbbabaabb$ is a de Brujin word of order 4 over the alphabet $\{a,b\}$. In fact one can verify that all strings of length 4 over $\{a,b\}$ appear as factor of $w$ just once.

Since we are here interested to linear and not to cyclic words, it is easy to verify that in order to have linear words containing all the $k$-length factors exactly once, it is sufficient to consider any linearization of the circular de Brujin word of order $k$ (that is, we cut the circular word in any position to get a linear word) and concatenate it with a word equal to its own prefix of length $k-1$. Therefore its length is $\sigma^k+k-1$.  We call such words {\em linear de Brujin sequences (or words)}.  For instance the linear de Brujin word corresponding to the circular one in the above example is the word  $w'=aaaababbbbabaabbaaa$.
Remark that the length of $w'$ is $2^k+k-1$.





In \cite{LZ1976} the following two theorems are proved.

\begin{theorem}\label{th-dBupper}
The number of phrases $c(n)$ in a LZ parsing of a sequence of length $n$ over an alphabet of size $\sigma$ satisfies:
$$c(n)\leq \frac{n}{(1-\epsilon_n)\log n}$$
where $\epsilon_n=2\frac{1+\log(\log (\sigma n))}{\log n}$.
\end{theorem}

\begin{theorem}\label{th-dBLower}
Let $B$ be a de Brujin sequence of order $k$ and length $n+k-1$ over an alphabet of size $\sigma$ ($n=\sigma^k$). Then $$c(B)\geq \frac{n}{\log n}$$
\end{theorem}
By combining Theorem \ref{th-dBupper}, \ref{th-dBLower} and Lemma \ref{L_KePrLZ}, we get both upper and lower bounds for a smallest string attractor for de Brujin sequences.

\begin{proposition}
Let $B$ be a de Brujin sequence of order $k$ and length $n+k-1$ over an alphabet of size $\sigma$ ($n=\sigma^k$). Then the cardinality $\gamma^*$ of a smallest string attractor for $B$ satisfies:
$$\frac{n}{\log n}\leq \gamma^* \leq \frac{n}{(1-\epsilon_n) \log n}+1$$
where $\epsilon_n=2\frac{1+\log(\log (\sigma n))}{\log n}$.
\end{proposition}

This means that $\gamma^*$ for a de Brujin word of length $n$ grows asintotically as $\frac{n}{\log n}$, corresponding to the worst case for the size of a smallest string attractor of any word in $\Sigma^*$.

Notice that the lower bound is somehow intuitively expected, since all the words of length $k$ appear only once in $B$, therefore two consecutive positions in any string attractor cannot be farthest than $k$.

For instance one can verify that for 
$w'=aaaababbbbabaabbaaa$ a smallest string attractor  is $\{4, 8, 12, 16 \}$.

\section{Conclusion and Open Problems}
In this paper we have studied the notion of string attractor from a combinatorial point of view. We have given an explicit construction of the string attractor for the infinite families of standard Sturmian words and Thue-Morse words. For standard Sturmian words, by using their combinatorial properties, the construction gives a smallest string attractor whose size is $2$. String attractors of minimum size can be also constructed for circularly balanced epistandard words. It is open the question to characterize all the words whose smallest string attractor has size equal to the cardinality of the alphabet. For Thue-Morse words, the size of the attractor is logarithmic with respect to the length of the word. We conjecture that such size is minimum and we leave it as an open problem.  
Based on the results presented in the paper two research directions could be explored. Some standard Sturmian words and Thue-Morse words are generated by morphisms. It could be interesting to find which properties of the morphism determine a smallest string attractor of constant size. Moreover, we plan to study how the distribution of the positions in the smallest string attractor is related to the combinatorial structure of the words.

Finally, the size of smallest string attractor could be used to define a new function to measure the complexity of infinite words. It could be interesting to investigate how such a measure is related with other known complexity measures, such as the factor complexity.

\section*{Acknowledgements}
S. Mantaci, G. Rosone and M. Sciortino are partially supported by the project MIUR-SIR CMACBioSeq (``Combinatorial methods for analysis and compression of biological sequences'') grant n.~RBSI146R5L.

\bibliographystyle{splncs04}
\bibliography{BWT}

\end{document}